\documentclass[authoryear,reqno,12pt,a4paper]{elsarticle}  
\usepackage{amsthm,amsmath,amssymb,setspace,tabularx,rotating}
\usepackage{svg} % 核心：支持SVG插图
\usepackage{orcidlink}
\usepackage{amsmath,natbib,graphicx,soul}

\usepackage[export]{adjustbox}
\usepackage{verbatim,booktabs}
\usepackage{pifont}
\usepackage{caption}
\usepackage{xcolor,hyperref}
\usepackage[toc,page]{appendix}

\usepackage{tabularx}
\usepackage{longtable}
\usepackage[ruled,vlined]{algorithm2e}
\usepackage[toc,page]{appendix}

 \usepackage[top=1in, bottom=1in, left=0.8in, right=0.8in]{geometry}
\usepackage{algpseudocode}
\usepackage{lscape}
\usepackage{pdflscape}

\hypersetup{
pdftitle={xxx},
pdfauthor={xxx},
pdfkeywords={xxx},
colorlinks=true,
linkcolor=blue,
citecolor=blue,
urlcolor=blue,
bookmarksnumbered=true,%Show chapter/section number in PDF thumbnails
pdfstartview=%Acrobat reader's default zoom setting (vs Fit)
}

\robustify\bfseries
\newtheorem{lemma}{Lemma}[section]

\newtheorem{remark}{Remark}[section]

\numberwithin{equation}{section}

\def\i{\mathrm{i}}
\def\dd{\,\mathrm{d}}

\newcommand{\citewiki}[1]{(\href{https://en.wikipedia.org/wiki/#1}{\textmd{\textsc{WikipediA}}})}

\begin{document}
\begin{frontmatter}
\renewcommand{\thefootnote}{\fnsymbol{footnote}}
\title{On options-driven realized volatility forecasting: Information gains \\
via rough volatility model
}

% \author[hku]{Yilin Chen}
% \ead{yilincyl@connect.hku.hk}

\author[ust,ustgz]{Zheqi Fan\orcidlink{0000-0003-4752-9019}}
\ead{zheqi.fan@connect.ust.hk}

\author[pine]{Meng Melody Wang}
% \ead{wangm@pinevc.com.cn}

\author[bnbu]{Yifan Ye\corref{corrauthor}\orcidlink{https://orcid.org/0009-0004-3110-4141}}
\ead{yifanye@bnbu.edu.cn}

% \address[hku]{HKU Business School,\\ 
% University of Hong Kong,
% Hong Kong, China}

\address[ust]{
Division of EMIA, Hong Kong University of Science and Technology, Hong Kong, China}
\address[ustgz]{
Thrust of FinTech, Hong Kong University of Science and Technology (Guangzhou), China}
\address[pine]{Green Pine Innorise Fund, Shenzhen, China}
\address[bnbu]{Faculty of Business and Management, Beijing Normal - Hong Kong Baptist University, Zhuhai, China}
% \address[bnbuai]{School of AI and Liberal Arts,
% Beijing Normal - Hong Kong Baptist University, Zhuhai, China}

\cortext[corrauthor]{Corresponding author %\textit{Email:} yyeam@connect.ust.hk, \textit{Address:} Thrust of Financial Technology, Society Hub, Hong Kong University of Science and Technology, Guangzhou, China
}
%-----------------------------------------
\date{\today}
%-----------------------------------------
\begin{abstract}
\small
% We augment the Heterogeneous Autoregressive Regression model for forecasting realized volatility by incorporating model-based volatility estimators computed from traded options data. 
% Specifically, we propose an iterative two-step approach to extract spot volatility under the rough stochastic volatility framework, leveraging the parametric inference framework of \cite{andersen2015parametric}. 
%  % We apply the Markovian approximation approach to overcome the computational issue of rough volatility model for pricing vanilla options.
% We apply deep learning surrogate technique to accelerate the process for extracting spot volatility estimators from large-scale options panel. 
% We demonstrate that these augmentations result in improved daily forecasts for realized volatility.
We examine whether model-based spot volatility estimators extracted from traded options data enhance the predictive power of the Heterogeneous Autoregressive (HAR) model for realized volatility. Specifically, we infer spot volatility under the rough stochastic volatility model via an iterative two-step approach following \cite{andersen2015parametric} and adopt a deep learning surrogate to accelerate model estimation from large-scale options panels. 
Benchmarked against traditional stochastic volatility models (Heston, Bates, SVCJ) and the VIX index, our results demonstrate that the augmented HAR-RV-RHeston model improves daily realized volatility forecasting accuracy and sustains superior performance across horizons up to one month.
\\~
\newline
%-----------------------------------------
~
\newline
%-----------------------------------------
\end{abstract}
\end{frontmatter}
%-----------------------------------------
\newpage
\section{Introduction}
\noindent 
Volatilities, intrinsically linked with macro- and microeconomics, play a central role in investments, asset pricing, risk management, and monetary policies. The past few decades have witnessed tremendous progress in modeling time-varying volatilities.
Thanks to the availability of high-frequency data and recent advancements in volatility measurement using such data, we can now estimate daily volatilities with high accuracy. The realized variance (RV)—the square of realized volatility—is a consistent estimator of the integrated variance (IV) as the sampling frequency increases, provided microstructure noise can be safely ignored \citep{barndorff2002econometric}.

\cite{corsi2009simple} proposed the Heterogeneous Autoregressive (HAR) model based on the heterogeneous market hypothesis, which effectively replicates the long-memory pattern of realized volatility. Despite its parsimonious structure, the HAR model has demonstrated robust forecasting performance and is now widely adopted as a benchmark in volatility forecasting literature. Subsequent research has introduced numerous extensions to enhance its predictive power, with a comprehensive review provided in \cite{clements2021practical}. Most existing volatility forecasting studies rely on backward-looking information to infer future volatility. In contrast, derivative prices contain forward-looking information, as they reflect market participants’ judgments about future cash flows and trends. Theoretically, incorporating such forward-looking information enriches the information set $\mathcal{F}_t$, making it unsurprising that methods leveraging this information often yield more accurate volatility predictions.

In the finance literature, \cite{christensen1998relation} was among the first to highlight that option prices, reflected in implied volatility—contain valuable information about future stock market volatility. Beyond their forward-looking nature, the nonlinear payoff structure of options inherently links their prices to the volatility of the underlying asset, endowing them with predictive potential for future volatility. Empirically, \cite{todorov2022information} showed that option data offers nontrivial gains for volatility forecasting, primarily by enabling more precise measurement of spot volatility. Building on this, \cite{michael2025options} extracted model-based volatility estimators from the Heston and Bates models and augmented the HAR framework, improving realized volatility forecasting performance.

A notable recent advancement in volatility modeling is the emergence of rough volatility models, which have garnered substantial attention in the financial industry and academic research in mathematical finance and financial econometrics since the seminal work of \cite{gatheral2018volatility}. Notably, the 2021 Risk Awards recognized the contributions of rough volatility models\footnote{\url{https://www.risk.net/awards/7736196/quants-of-the-year-jim-gatheral-and-mathieu-rosenbaum}} for their ability to capture refined volatility dynamics (e.g., the rough path persistence of volatility) that traditional models overlook. However, rough volatility models are non-Markovian, lacking closed-form solutions and introducing increased computational complexity. A core model in this class is the rough Heston model \citep{el2019characteristic}, which features an explicit characteristic function and serves as a natural extension of traditional stochastic volatility model.

Against this backdrop, we augment the HAR model by incorporating model-based spot volatility estimators inferred from options data. Specifically, we use the spot volatility estimator from the rough Heston model as the core extended predictor, with three traditional stochastic volatility models—Heston \citep{Heston1993}, Bates \citep{Bates2000}, and SVCJ \citep{eraker2003impact}—as benchmarks. We also include the non-parametric option-implied volatility index (VIX) as an exogenous predictor, forming the HAR-RV-VIX model which serves as an additional benchmark.

To infer the spot volatility series from each stochastic volatility model, we adopt the parametric inference framework proposed by \cite{andersen2015parametric}, which involves an iterative two-step procedure. First, we initialize the model’s structural parameters and derive the corresponding spot variance series. Second, we optimize these structural parameters using the spot variances obtained in the first step. This process iterates until a predefined termination criterion is satisfied, ensuring consistent inference of spot volatility. Given the computational complexity of rough volatility models in large options panel settings, we leverage deep learning techniques to approximate the option pricing function as a lightweight surrogate, accelerating the estimation process.

Using the inferred spot volatility series as new predictors, we conduct comprehensive empirical analyses: in-sample fitting, out-of-sample realized volatility forecasting, and multi-horizon predictive power tests. Our empirical results demonstrate that the augmented HAR-RV-RHeston model outperforms the benchmark models, delivering improved forecasting accuracy for daily realized volatility and sustaining superior performance even for forecast horizons up to one month.

~

\noindent\textbf{Related literature} -  Our study contributes to the extensive body of literature on realized volatility (RV) forecasting. 
While a wealth of research has focused on modeling and predicting daily return volatility, the majority of existing approaches—predominantly parametric GARCH or stochastic volatility (SV) models—derive daily volatility forecasts exclusively from daily return data. 
A key limitation of these traditional frameworks is their inability to leverage high-frequency data, which has become increasingly accessible in recent years. 
Against this backdrop, RV (computed as the sum of squared intraday returns) has emerged as a widely adopted measure of volatility, addressing the shortcomings of daily return-based models.
\cite{corsi2009simple} proposed a parsimonious autoregressive (AR)-type model, known as the Heterogeneous Autoregressive (HAR) model, to forecast daily RV by decomposing it into volatility components across short-, medium-, and long-term time horizons. 
Subsequent studies have extended the baseline HAR framework by incorporating additional predictive terms \citep{haugom2014forecasting,jiang2019volatility,bekaert2014vix,ruan2025merton}. 
Relative to this line of research, our paper seeks to enhance the HAR model by integrating novel informational variables derived from options data. 
We adopt a model-based approach: specifically, we apply the parametric inference framework of \cite{andersen2015parametric} to extract spot volatility estimators from a rough stochastic volatility (RSV) model. 

Among these, our paper is closely related to \cite{michael2025options} but not a mere extension. We advance the literature through three distinct, logically connected improvements: we adopt a more sophisticated rough volatility model that captures more realistic volatility dynamics consistent with empirical evidence from both realized volatility time series \citep{gatheral2018volatility} and options markets \citep{bayer2016pricing}; we use a more robust parametric inference framework, with the consistency and finite sample properties of focal estimators formally established in \cite{andersen2015parametric}, to extract spot volatility estimators from large option panels, ensuring reliable recovery of latent states and model parameters; and we address the computational challenge of pricing under rough volatility models (which go beyond the affine jump-diffusion class \cite{duffie2000}) by leveraging deep learning to train a lightweight surrogate for the option pricing function, enabling efficient estimation on large-scale datasets.

Our paper also contributes to the emerging and increasingly popular literature on machine learning applications in finance. Most existing studies in this strand focus on methodological innovations, such as employing machine learning, deep learning or agentic AI approaches to forecast asset returns or realized volatility \citep{zhang2024volatility,zhang2025forecasting,li2025automated,huang2026beyond}. Our approach to adopting machine learning differs from these studies: rather than directly training neural networks on market data, we train a deep learning model to learn the mapping from synthetic data, leveraging the function approximation property of neural networks. We then utilize this lightweight statistical surrogate to accelerate the intensive computations involved in subsequent model estimation and calibration. This approach aligns more closely with the literature on deep learning-based calibration \citep{Blanka2021,rosenbaum2021deep,Chen2026,fan2026deep}.

The remainder of this paper is structured as follows. 
Section \ref{sec:data} formally defines realized volatility and provides a detailed description of the dataset employed in our analysis. 
In Section \ref{sec:extract}, we elaborate on the methodological framework for extracting the spot volatility estimator from options panel data, based on the (rough) stochastic volatility model. 
Section \ref{sec:empirics} presents and discusses the empirical findings derived from our analysis. 
Finally, Section \ref{section-conclusion} summarizes the key findings of this study and offers concluding remarks.

\section{Data and Realized Volatility}
\label{sec:data}

\subsection{Realized volatility}
\noindent We use high-frequency intraday returns for S\&P 500 from Thomas Reuters Tick History database.
Our data sample covers the time period from January 2011 to June 2021, with June 2021 to June 2021 as the out-of-sample testing period.
Regarding the sampling frequency, we use 5-minute log-returns for which the market microstructure noise can be safely ignored \citep{liu2015does}.
We also use daily close values of the Cboe volatility index (VIX)
related to the S\&P 500 for candidate models.
The options data we used for extracting spot volatility estimator are retrieved from OptionMetrics, and the details are outline in Appendix \ref{sec:options data}.

\begin{table}[htbp]
\caption{\textbf{Descriptive statistics for RV values}}
\par
{\footnotesize
This Table shows the summary statistics of the SPX realized volatility computed from intraday 5-min returns. 
}
\noindent
\begin{center}
{\scriptsize
\begin{tabular}{rrrrrrrrrr}
\toprule
           &       mean &        std &        min &       25\% &       50\% &       75\% &        max &       skew &       kurt \\
\midrule

 $RV_t^d$ &    0.0815  &    0.0513  &    0.0150  &    0.0494  &    0.0680  &    0.0972  &    0.5797  &    2.5970  &   11.6092  \\

 $RV_t^w$ &    0.0815  &    0.0450  &    0.0224  &    0.0529  &    0.0696  &    0.0933  &    0.4147  &    2.2750  &    7.6920  \\

$RV_t^m$ &    0.0816  &    0.0379  &    0.0286  &    0.0583  &    0.0700  &    0.0908  &    0.2354  &    1.6389  &    2.5747  \\
\bottomrule
\end{tabular}        
   }
\end{center}
\label{tab:summary rv}
\end{table}

\begin{figure}[htbp]
\begin{center}
\includegraphics[width=0.80\textwidth]{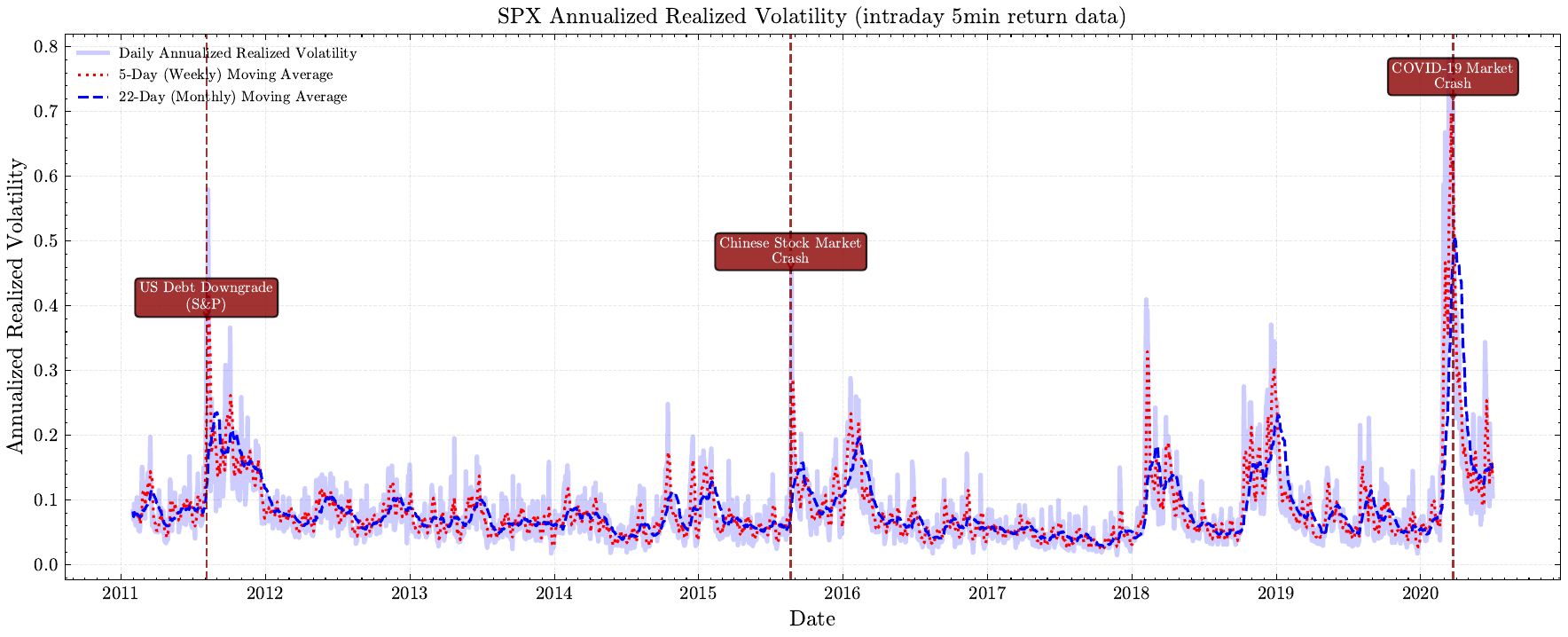}
\end{center}
\caption{\textbf{Realized volatility from high frequency returns}}
{
\footnotesize
\noindent
\begin{spacing}{1.4}
The figure displays the daily realized volatility of SPX computed from the high-frequency intraday returns.
\end{spacing}
}
\label{fig:ts-rv}
\end{figure}

% \subsection{RV alternative}

In general, let \(P_{i,t}\) denote the price process of financial asset \(i\), which follows the stochastic differential equation:
\begin{equation}
\mathrm{d} \log P_{i,t} = \mu_i \mathrm{d}t + \sigma_{i,t} \mathrm{d}W_t,
\end{equation}
where \(\mu_i\) is the constant drift term, \(\sigma_{i,t}\) represents the time-varying instantaneous volatility, and \(W_t\) is a standard Brownian motion. The theoretical integrated variance (IV) of asset \(i\) over the time interval \((t-h, t]\) is defined as:
\begin{equation}
\operatorname{IV}_{i,t}(h) = \int_{t-h}^t \sigma_{i,s}^2 \mathrm{d}s,
\end{equation}
with \(h\) denoting the look-back horizon (e.g., 1 trading day).

To estimate the unobservable integrated variance, we use high-frequency intraday data and adopt 5-minute logarithmic mid-price returns. As demonstrated by \cite{liu2015does}, the 5-minute sampling interval balances the trade-off between mitigating microstructure noise \citep{hansen2006realized} and retaining informative volatility dynamics, outperforming other sub-sampling frequencies for daily realized volatility (RV) forecasting and thus becoming a widely accepted standard in the literature. The 5-minute logarithmic return of asset \(i\) during \((t-1, t]\) is calculated as:
\begin{equation}
r_{i,t} := \log\left(\frac{P_{i,t}}{P_{i,t-1}}\right),
\end{equation}
where \(P_{i,t}\) is the mid-price at time \(t\), defined as \(P_{i,t} = \frac{P_{i,t}^b + P_{i,t}^a}{2}\). Here, \(P_{i,t}^b\) and \(P_{i,t}^a\) represent the best bid price and best ask price of asset \(i\) at time \(t\), respectively.

\cite{andersen2001distribution} and \cite{barndorff2002econometric} have proven that the sum of squared intraday returns is a consistent estimator of the unobservable integrated variance. Given the accessibility of high-frequency data, we use realized volatility (RV) as a proxy for IV. Notably, daily realized variance exhibits a highly positive skewed distribution. To alleviate the impact of extreme values and improve model fitting, we follow \citep{zhang2024volatility, michael2025options} among others and apply a logarithmic transformation. Specifically, the logarithmic realized volatility of asset \(i\) over the look-back horizon \(b\) (i.e., during \((t-b, t]\)) is defined as:
\begin{equation}
\mathrm{RV}_{i,t}^{(b)} := \log\left[\sum_{s=t-b+1}^t r_{i,s}^2\right].
\end{equation}

Figure \ref{fig:ts-rv} illustrates the time series of computed (annualized) realized volatility over our sample periods. Table \ref{tab:summary rv} reports the descriptive statistics of the (annualized) daily, weekly, and monthly RVs without log transformation. These series are
significantly skewed.

\section{Extracting volatility estimators under stochastic volatility models}
\label{sec:extract}
\noindent In this section, we describe how to extract the spot volatility estimators from large panel of traded options prices under the rough Heston model \eqref{eq-rheston} and three widely-used stochastic volatility models Heston, Bates, and SVCJ model\footnote{Please find more details in Appendix \ref{appendix:model details}.}.

\subsection{Rough Heston stochastic volatility model}
\noindent On a complete probability space $(\Omega, \mathcal{F}, \mathbb{Q})$ with a filtration $(\mathcal{F}_t)_{t \in [0,T]}$ and a risk-neutral measure $\mathbb{Q}$, the rough Heston model proposed in \cite{el2018perfect} and \cite{el2019characteristic} is represented as:
\begin{equation}
\label{eq-rheston}
\begin{aligned}
\frac{\dd S_t}{S_t}=r\dd t+ \sqrt{V_t}\dd B_t,\ V_t=g_0(t)+\int_0^t K(t-s)\left(-\lambda V_s\dd s+\nu \sqrt{V_s} \dd W_s \right),
\end{aligned}
\end{equation}
where $B_t$ and $W_t$ are a pair of correlated standard Brownian motions with $\dd B_t  \dd W_t=\rho \dd t$. 
Here, the correlation coefficient $\rho$ and all other parameters are assumed to be constant. Let $r$ denote the risk-free interest rate. We note that this is of the class of one-factor models and we will compare this to a family of one-factor stochastic volatility models. In practice, researchers typically employ a kernel approximation of the form \citep{abi2019markovian,abi2019multifactor},
$
K^n(t) = \sum_{i=1}^n c_i e^{- x_i t }, \quad  t\geq 0$.
This approximation significantly simplifies the numerical implementation of the rough Heston model by exploiting the Markovian structure of the approximating model. It should be noted, however, that even under the Markovian approximation introduced by \cite{AbiJaber2019}, simulating the model remains computationally expensive. Such kernel approximation leads to the following LHeston model proposed in \cite{abi2019multifactor} and \cite{AbiJaber2019}. Under the same risk-neutral pricing measure $\mathbb{Q}$, the joint dynamics of the index value process $S_t$ and its instantaneous variance $V_t$ takes the form
\begin{equation}
\label{eq-lheston-cj}
\begin{aligned}
&\frac{\dd S_t}{S_t}=r \dd t+ \sqrt{V_t} \dd B_t, \\
&\dd V_t=\dd g_0^n(t)+\sum_{i=1}^{n} c_{i} \dd U_t^{i}, \\
&\dd U_t^{i}=\left(-x_i U_t^{i}-\lambda V_t\right) \dd t+\nu \sqrt{V_t} \dd W_t,\ U_0^{i}=0, \quad i=1, \ldots, n ,
\end{aligned}
\end{equation}

All the variance components $\left(U^{i}\right)_{1 \leq i \leq n}$ start from zero and share the same one-dimensional Brownian motion $W$, except that they exhibit mean reversion at different rates $\left(x_i\right)_{1 \leq i \leq n}$. 
In addition, they share the common negative drift rate $-\lambda V_t$.
The variance components are characterized by the weight coefficients $\left(c_i\right)_{1 \leq i \leq n}$.
% The deterministic function $g_0^n$ can be chosen by fitting initial variance curves. 
We follow \cite{AbiJaber2019} to choose
$
g_0^n(t) = V_0 + \lambda \theta \sum_{i=1}^n c_i \int_0^t e^{-x_i(t-l)} \dd l=V_0+\lambda \theta \sum_{i=1}^n \frac{c_i}{x_i}\left(1-e^{-x_i t}\right)
$. The characteristic functions of rough Heston \eqref{eq-rheston} used for computing option prices are detailed in Appendix \ref{sec:ChF}.
%-------------------------------------------

%-------------------------------------------
\subsubsection{Deep surrogate for vanilla option pricing}
\noindent While the characteristic function of the rough Heston model \eqref{eq-rheston} is analytically available, it involves a fractional Riccati equation. This leads to substantial computational inefficiency when applying Fourier inversion techniques for European option pricing. For its multi-factor Markovian approximation—the lifted Heston model \eqref{eq-lheston-cj}—the characteristic function still requires solving dozens of Riccati ordinary differential equations (ODEs), which remains computationally intensive.

We denote the option price obtained via numerical methods (e.g., Monte Carlo simulation method or Fourier inversion method via (semi-)explicit characteristic function) as \(P^{Model}(\Theta;\xi)\), where \(\Theta\) represents model-specific parameters and \(\xi\) denotes contract-related variables (e.g., strike price, time to maturity). Corresponding market prices are denoted as \(P^{Mkt}(\xi)\). To accelerate model estimation and calibration, we adopt deep learning techniques following \citep{Blanka2021, rosenbaum2021deep, Chen2026,fan2026deep}, specifically the ``from model parameters to prices (MtP)" approach. This framework consists of two core steps: first, training deep neural networks to approximate the pricing function—mapping model parameters and contract variables to option prices; second, applying traditional optimization algorithms to find optimal parameters that minimize the discrepancy between model outputs and market data.

Benefiting from the multi-factor Markovian approximation of rough Heston model, which is more analytically tractable, we compute derivative prices under the lifted Heston framework via Fourier inversion.
We construct a large synthetic dataset (with millions of samples) where each observation \((X, y)\) includes combinations of \(\{\Theta, \xi\}\) as inputs \(X\), and the corresponding option price \(P^{Model}(\Theta;\xi)\) as the target \(y\). The neural networks employed are multilayer perceptrons, with hyperparameters tuned following established literature. The theoretical validity of such function approximation is guaranteed by the Universal Approximation Theorem.

After training, the computationally lightweight deep surrogate replaces the costly numerical pricing engine in the estimation framework. For detailed technical implementations, we refer to the aforementioned studies. To verify the approximation accuracy of the pre-trained surrogate, we evaluate average percentage errors and their standard deviations across sub-modules of moneyness and maturity, as presented in Figure \ref{fig:surrogate accuracy}. The results confirm that the neural network achieves reliable approximation performance.

\begin{figure}[htbp]
\begin{center}
\includegraphics[width=0.80\textwidth]{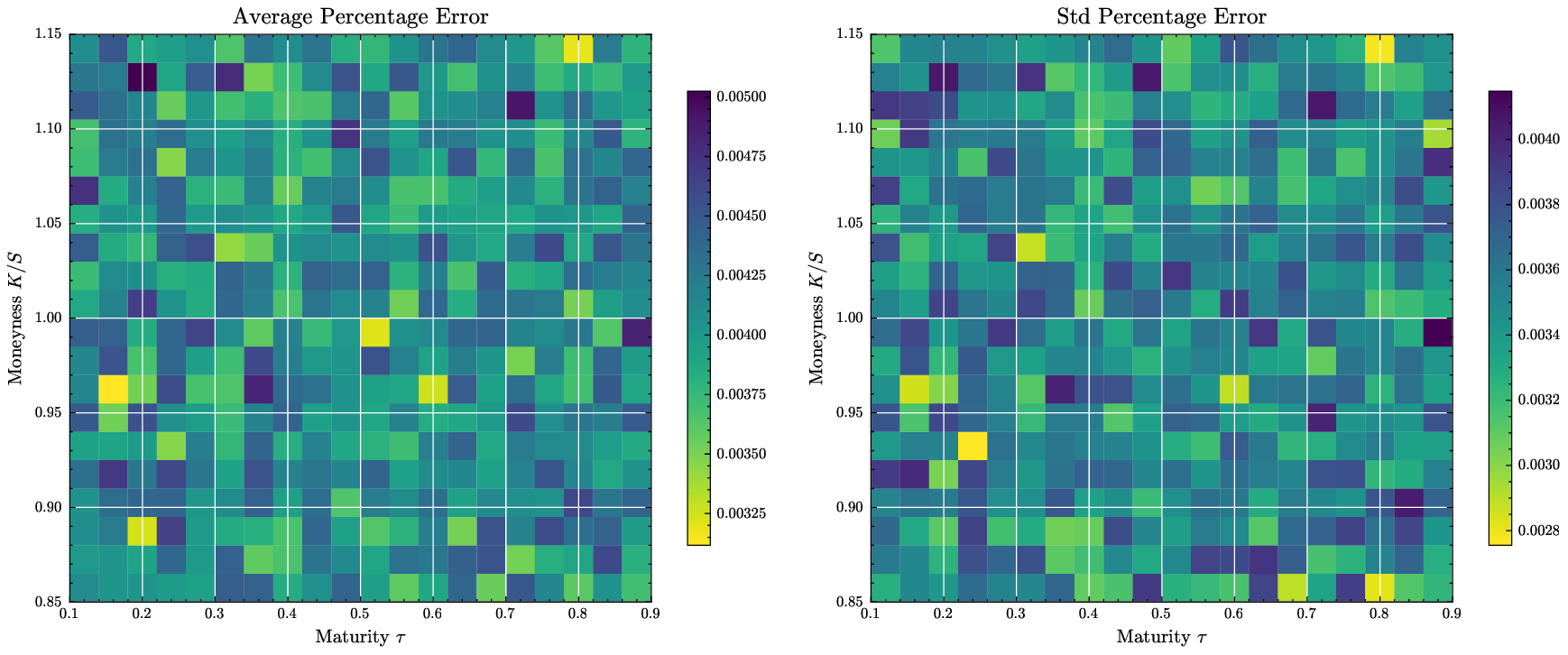}
\end{center}
\caption{\textbf{Accuracy for pre-trained deep surrogate}}
{
\footnotesize
\noindent
\begin{spacing}{1.4}
The figure displays the numerical accuracy of our pre-trained deep learning surrogate for option pricing engine.
\end{spacing}
}
\label{fig:surrogate accuracy}
\end{figure}

\subsection{Parametric inference}

\noindent We adopt the parametric inference framework for option panels and latent state recovery proposed by \citet{andersen2015parametric}, which has been rigorously validated via asymptotic theory (consistency, stable convergence to mixed-Gaussian laws) and extensive Monte Carlo simulations. This method is particularly suited for extracting the single latent volatility $V_t$ and model parameters from option panels with fixed time spans and expanding cross-sections, an advantage we leverage in our empirical implementation, following the application logic of \citet{andersen2015risk}.

The core estimation hinges on solving a penalized nonlinear least squares (NLS) optimization, integrating cross-sectional option price variation and time-series model-free realized volatility from high-frequency underlying asset data. For our single latent variable $V_t$, the optimization problem is formally defined as:
\begin{equation}
\begin{aligned}
\left(\left\{\hat{V}_t\right\}_{t=1, \ldots, T}, \hat{\theta}\right)= \underset{\left\{V_t\right\}_{t=1, \ldots, T}, \theta \in \Theta}{\arg \min } \sum_{t=1}^T\left\{
\frac{1}{N_t} \sum_{j=1}^{N_t}\left(\tilde{\kappa}_{t,k_j,\tau_j} - \kappa\left(k_j, \tau_j, V_t, \theta\right)\right)^2 + 
\lambda_n \left(\sqrt{\hat{V}_t^n} - \sqrt{V_t}\right)^2
\right\},
\label{eq:andersen_core_est}
\end{aligned}
\end{equation}
where $\tilde{P}_{t,k_j,\tau_j}$ denotes the market-observed option price for the $j$-th option on day $t$, characterized by moneyness $k_j = K/S_t$ (with $K$ as strike price and $S_{t}$ as underlying price) and time-to-maturity $\tau_j$;
$P\left(k_j, \tau_j, V_t, \theta\right)$ is the model-implied option price, derived from the risk-neutral dynamics parameterized by $\theta$ (full parameter vector for our model) and the single latent spot volatility $V_t$;
$\hat{V}_t^n$ is the nonparametric realized volatility estimator.
$\lambda_n$ is the penalization weight balancing option fit and alignment with high-frequency volatility. Following \cite{andersen2015parametric,andersen2015risk}, we set $\lambda_n=0.2$, ensuring the option panel dominates estimation while regularizing extreme $\hat{V}_t$. 
$N_t$ is the number of options observed on day $t$.\footnote{The key difference between the parameters and the
state vector is that the latter changes from day to day, while the former must
remain invariant across the sample. The longer the time span covered by the
sample, the more restrictive is this invariance condition for the risk-neutral
measure.\citep{andersen2015parametric}.}
\begin{remark}
    The objective function \eqref{eq:andersen_core_est} includes two parts. The first component is the mean squared error of implied volatilities derived from a panel data of market option prices and option prices computed from a model-based option pricing formula, respectively.
    \begin{equation}
{\operatorname{Option~Fit}_t}=\frac{1}{N_t} \sum_{j=1}^{N_t}\left(\hat{P}_{t, k_j, \tau_j}-P\left(k_j, \tau_j, V_t, \theta\right)\right)^2,
\end{equation}
    The second part of the objective function punishes the estimation according to the degree of deviation between a time-series of inferred spot volatilities and realized volatilities. 
    \begin{equation}
    \label{eq：vol fit}
\operatorname{Vol~Fit}_t=\left(\sqrt{\widehat{V}_t^{\left(n, m_n\right)}}-\sqrt{V\left( \theta\right)}\right)^2
\end{equation}
This
restriction follows from the fact that the diffusion coefficient of X is invariant
under an equivalent measure change (from $\mathbb{Q} $ to $\mathbb{P}$), so the two estimates should not be statistically distinct.
\end{remark}

The estimator specified in \eqref{eq:andersen_core_est} is derived through joint optimization over the model parameter vector $\theta$ and the sequence of latent state vector realizations $\{S_t\}_{t=1}^T$. Despite the inherent high-dimensionality of this optimization problem (involving both cross-sectional parameters and time-varying latent states), its tractability is guaranteed by the specific structure of the objective function, as emphasized in \citet{andersen2015risk}. An illustrative workflow to help understand the implementation details is shown in Figure \ref{fig:flowchart}.

\begin{figure}[htbp]
\begin{center}
\includegraphics[width=0.99\linewidth]{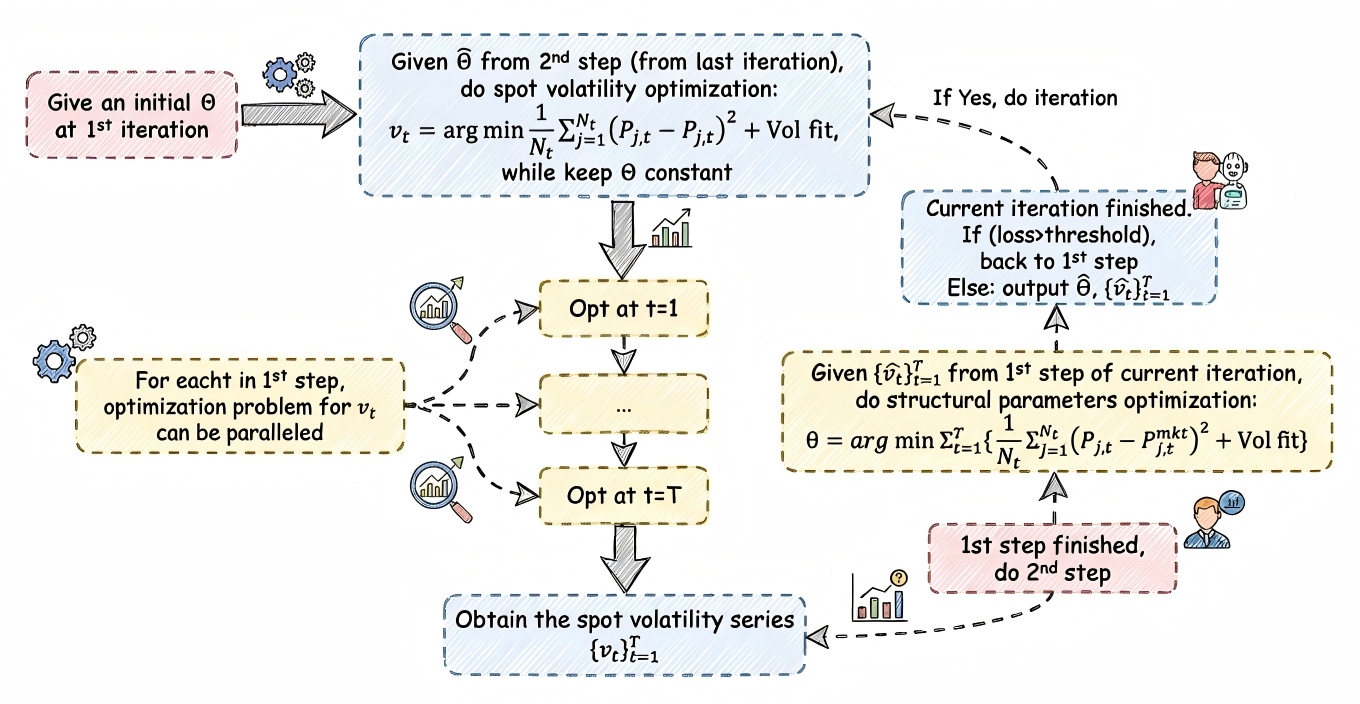}
\end{center}
\caption{\textbf{Flowchart of the parametric inference procedure}}
{
\footnotesize
\noindent
\begin{spacing}{1.4}
The figure illustrates in detail the process of the iterative two-step estimation procedure to extract of the volatility estimator.
\end{spacing}
}
\label{fig:flowchart}
\end{figure}

\subsubsection{Implementation details of iterative two-step procedure}  
\noindent This iterative two-step procedure treats the spot volatilities as latent variables that are re-estimated in a daily basis.\footnote{It is observed that, in practical applications, financial model parameters are generally divided into static modeling choices and frequently updated market inputs. To maintain market consistency while treating the latent spot volatility as a dynamic market parameter, a two-step calibration procedure \citep{ballotta2022smiles,fan2026deep,fan2026jod}.} Indeed, we estimate the model’s structural parameter set $\Theta = \left\{ \kappa,\theta,\sigma,\mu,\rho,H \right\}$, as well as the spot variance $\left\{V_t\right\}_{t=1, \ldots, T}$ separately, where $T$ is the total number of trading days used in the estimation period. The details of the iterative procedure are presented as follows.

\begin{enumerate}
\item \textbf{Inferring spot variance}. We start with a given set of structural parameters $\Theta$. For each date $t=1,...,T$ within the estimation sample, we solve a nonlinear least squared optimization to get the spot variance series:
\begin{equation}
\begin{aligned}
\begin{aligned}
\left\{\widehat{V}_{t}\right\} =\arg \min \frac{1}{N_t}\sum_{j=1}^{N_t}\left(C_{j, t}-C_j\left(\Theta, V_{t}\right)\right)^2  / \operatorname{Vega}_{j, t}^2 + \operatorname{Vol\ Fit}_t
\end{aligned}
\end{aligned}
\label{eq:objective}
\end{equation}
where $C_{j, t}$ is the market price of contract $j$ on date $t$, and $C_j\left(\Theta, V_{t}\right)$ is the corresponding model-based price. $N_t$ is the number of contracts available on a day $t$. $\operatorname{Vega}_{j, t}$ is the Black-Scholes vega. The last penalty term is given by \eqref{eq：vol fit}.

\item \textbf{Structural parameters estimation}. The second step is to optimize the aggregate objective function using the estimated spot variance $\left\{\widehat{V}_{t}\right\}, t=1,...,T$ obtained from step 1 of the same iteration:
\begin{equation}
\begin{aligned}
\begin{aligned}
\widehat{\Theta}=\arg \min \frac{1}{N} \sum_{j, t}^N\left(C_{j, t}-C_j\left(\Theta, V_{t}\right)\right)^2 / \operatorname{Vega}_{j,t}^2
\end{aligned}
\end{aligned}
\label{eq:objective2}
\end{equation}
with all data points $N=\sum_{t=1}^T N_t$.
\end{enumerate}
The procedure iterates between Step 1 and Step 2 until no further significant decreases in the overall objective \eqref{eq:objective} in Step 2 are obtained. 
Here we use the widely-used first-order Taylor approximation to the implied volatility errors $
\text { IVRMSE }  \equiv \sqrt{\frac{1}{N} \sum_{j, t}^N\left(\sigma_{j, t}-\sigma_{j, t}\left(\Theta, V_{t}\right)\right)^2} 
\approx \sqrt{\frac{1}{N} \sum_{j, t}^N\left(C_{j, t}-C_j\left(\Theta, V_{t}\right)\right)^2 / \operatorname{Vega}_{j, t}^2}$.
To avoid look-ahead bias when inferring the spot volatility series for future realized volatility forecasting, we adopt an out-of-sample framework following \citep{christoffersen2009shape,ye2025modeling,ye2026vix}: we estimate model parameters annually and keep these structural parameters constant to infer the spot volatility for the subsequent year. 
Figure \ref{fig:spot vol} presents the time series of inferred spot volatility derived from the Heston, Bates, SVCJ, and rough Heston models. 
Table \ref{tab:spot vol} reports the corresponding summary statistics for these series. They are also quite positively skewed, and we will take the log transformation when forecasting realized volatilities.

\begin{figure}[htbp]
\begin{center}
\includegraphics[width=1.00\textwidth]{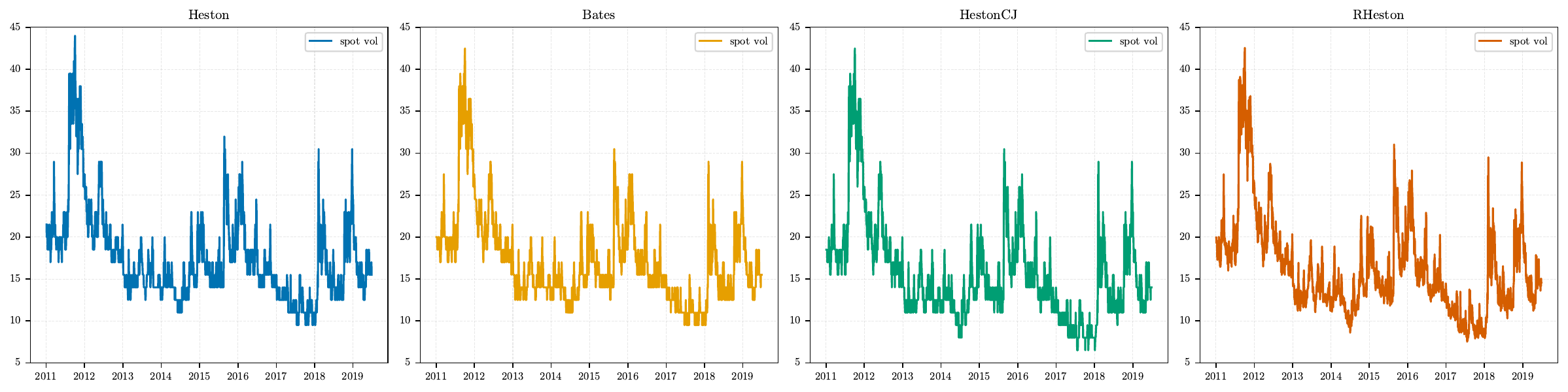}
\end{center}
\caption{\textbf{Spot volatility estimators from SV model}}
{
\footnotesize
\noindent
\begin{spacing}{1.4}
The figure displays the daily realized volatility of SPX computed from the high-frequency intraday returns. The daily spot prices
are also shown here in the black dotted line to help benchmark.
\end{spacing}
}
\label{fig:spot vol}
\end{figure}

\begin{table}[htbp]
\caption{\textbf{Summary statistics for spot vol estimator}}
\par
{\footnotesize
This table shows the choice of hyperparameters about the neural network architecture we use. 
}
\noindent
\begin{center}
{\scriptsize
\begin{tabular}{rrrrrrrrrr}
\toprule
           &       mean &        std &        min &       25\% &       50\% &       75\% &        max &       skew &       kurt \\
\midrule
    Heston &    0.1763  &    0.0588  &    0.0565  &    0.1371  &    0.1657  &    0.2021  &    0.4487  &    1.3589  &    2.7458  \\

     Bates &    0.1716  &    0.0564  &    0.0608  &    0.1328  &    0.1631  &    0.1965  &    0.4642  &    1.2804  &    2.5237  \\

  SVCJ &    0.1559  &    0.0587  &    0.0494  &    0.1163  &    0.1445  &    0.1830  &    0.4328  &    1.4435  &    2.8361  \\

   RHeston &    0.1624  &    0.0566  &    0.0752  &    0.1246  &    0.1494  &    0.1870  &    0.4257  &    1.5261  &    3.0281  \\
\bottomrule
\end{tabular}      
   }
\end{center}
\label{tab:spot vol}
\end{table}

\section{Empirical results}
\label{sec:empirics}
\subsection{HAR and HARX Models}
\noindent We adopt the Heterogeneous Autoregressive (HAR) model proposed by \cite{corsi2009simple} as the benchmark for realized volatility (RV) modeling and forecasting. Inspired by the Heterogeneous Market Hypothesis, the HAR model captures volatility dynamics driven by market participants with distinct trading horizons (short-term, medium-term, long-term). It approximates long-memory characteristics through an additive cascade of volatility components corresponding to these horizons, avoiding complex long-memory specifications while maintaining predictive power.

First, we define horizon-specific volatility measures as the average of daily RVs over respective periods:
$
\mathrm{RV}_{t-1}^{(h)} \equiv \frac{1}{h} \sum_{j=1}^h \mathrm{RV}_{t-j},
$
where \(h \in \{1,5,22\}\) represents daily, weekly, and monthly horizons consistent with typical trading calendars. The baseline HAR model is specified as:
\begin{equation}
\operatorname{HAR}: \log(\mathrm{RV}_t) = c + \beta^{(1)} \log\left(\mathrm{RV}_{t-1}^{(1)}\right) + \beta^{(5)} \log\left(\mathrm{RV}_{t-1}^{(5)}\right) + \beta^{(22)} \log\left(\mathrm{RV}_{t-1}^{(22)}\right) + \varepsilon_t
\end{equation}
where \(c\) is the constant term, \(\beta^{(\cdot)}\) are coefficients for horizon-specific volatility components, and \(\varepsilon_t\) denotes the innovation term. We use the logarithmic transformation of RV to address its high positive skewness, improving model fit.
Model parameters are estimated via ordinary least squares (OLS). To address potential heteroscedasticity, we apply the Newey-West covariance correction \citep{newey1987hypothesis}. 
\\

Various extensions of the HAR model incorporating exogenous variables (HARX models) have been proposed in the vast literature and shown to enhance predictive performance \citep{bekaert2014vix, jiang2019volatility,todorov2022information,michael2025options}.To assess the incremental information of option-implied both non-parametric volatility (VIX) and model-based spot volatility estimators from stochastic volatility models, we specify the following HARX models:

The HAR-RV-VIX model \citep{bekaert2014vix,michael2025options,ruan2025merton} incorporates VIX as an exogenous variable to capture option-implied volatility information:
\begin{equation}
\begin{aligned}
\log(\mathrm{RV}_t) = c & + \beta^{(1)} \log\left(\mathrm{RV}_{t-1}^{(1)}\right) + \beta^{(5)} \log\left(\mathrm{RV}_{t-1}^{(5)}\right) + \beta^{(22)} \log\left(\mathrm{RV}_{t-1}^{(22)}\right)  + \beta^{\text{VIX}} \text{VIX}_{t-1} + \varepsilon_t
\end{aligned}
\end{equation}

The HAR-RV-SV model augments the baseline with SV estimators inferred from different stochastic volatility models:
\begin{equation}
\begin{aligned}
\log(\mathrm{RV}_t) = c & + \beta^{(1)} \log\left(\mathrm{RV}_{t-1}^{(1)}\right) + \beta^{(5)} \log\left(\mathrm{RV}_{t-1}^{(5)}\right) + \beta^{(22)} \log\left(\mathrm{RV}_{t-1}^{(22)}\right)  + \beta^{\text{SV}} \text{SV}_{t-1}^{\mathcal{M}_i} + \varepsilon_t
\end{aligned}
\end{equation}
where \(\mathcal{M}_i\) denotes the set of stochastic volatility models including Heston, Bates, SVCJ, and rough Heston, and \(\text{SV}_{t-1}^{\mathcal{M}_i}\) is the spot volatility estimator from model \(\mathcal{M}_i\) at time \(t-1\). Hereafter, we will call them like HAR-RV-RHeston.

\subsection{Forecast evaluation}
\label{sec:forecast_evaluation}

\noindent To evaluate the forecast accuracy of competing models, we adopt several widely used loss functions for volatility forecasting, including mean squared error (MSE), mean absolute error (MAE). These metrics are given as:

\begin{equation}
\begin{aligned}
 \mathrm{MSE} = \frac{1}{T} \sum_{t=1}^T \left( \widehat{\mathrm{RV}}_t - \mathrm{RV}_t \right)^2,\ 
 \mathrm{MAE} = \frac{1}{T} \sum_{t=1}^T \left| \widehat{\mathrm{RV}}_t - \mathrm{RV}_t \right|, \\
 % \mathrm{MAPE} &= \frac{1}{T} \sum_{t=1}^T \frac{\left| \widehat{\mathrm{RV}}_t - \mathrm{RV}_t \right|}{\mathrm{RV}_t},
\end{aligned}
\end{equation}
where \( T \) denotes the number of trading days in the evaluation sample, \( \mathrm{RV}_t \) is the daily realized volatility on day \( t \), and \( \widehat{\mathrm{RV}}_t \) is the corresponding forecast of \( \mathrm{RV}_t \) generated by the model. In addition to the above metrics, we also compute the Quasi-Likelihood (QLIKE), which belongs to the family of loss functions robust
to a noisy volatility proxy and is particularly suitable for volatility forecasting \citep{patton2011volatility} and defined as:

\begin{equation}
\text{QLIKE} = \frac{1}{N} \sum_{i=1}^N \frac{1}{ {T}_{}} \sum_{t \in {T}} \left[ \frac{\exp\left(\mathrm{RV}_{i,t}^{(b)}\right)}{\exp\left(\widehat{\mathrm{RV}}_{i,t}^{(b)}\right)} - \left( \mathrm{RV}_{i,t}^{(b)} - \widehat{\mathrm{RV}}_{i,t}^{(b)} \right) - 1 \right],
\end{equation}
where \( \widehat{\mathrm{RV}}_{i,t}^{(b)} \) denotes the predicted value of \( \mathrm{RV}_{i,t}^{(b)} \) (the realized volatility of stock \( i \) over the interval \( (t-b, t] \)), \( N \) is the number of stocks in our sample universe, \( \mathcal{T}_{\text{test}} \) represents the testing period, and \( \# \mathcal{T}_{\text{test}} \) is the length (number of trading days) of the testing period.

In practical applications such as trading strategy optimization, the accuracy of volatility level forecasts is not the sole focus—reliably predicting the \textit{direction} of volatility changes (i.e., whether volatility will rise or fall relative to the previous period) is equally critical. To assess this directional predictive power, we further evaluate the mean directional accuracy (MDA) of each model, defined as:

\begin{equation}
\mathrm{MDA} = \frac{1}{T} \sum_{t=1}^T \mathbf{1} \left\{ \operatorname{sgn}\left( \widehat{\mathrm{RV}}_t - \mathrm{RV}_{t-1} \right) = \operatorname{sgn}\left( \mathrm{RV}_t - \mathrm{RV}_{t-1} \right) \right\}
\end{equation}
where \( \mathbf{1}\{\cdot\} \) is the indicator function (taking a value of 1 if the enclosed condition is satisfied and 0 otherwise), and \( \operatorname{sgn}(\cdot) \) is the sign function (returning 1 for positive arguments, -1 for negative arguments, and 0 for zero). MDA quantifies the proportion of trading days where the forecasted direction of volatility change aligns with the actual direction, with values closer to 1 indicating superior directional predictive performance.

To assess the statistical significance of forecast accuracy differences across models, we employ the Diebold-Mariano (DM) test \citep{diebold2002comparing} for distinguishing the forecasting performance of time-series models. Let \( L(e_t) \) denote the loss function associated with forecast error \( e_t \) (e.g., \( L(e_t) = |e_t| \)); the loss difference between the forecasts of models \( a \) and \( b \) is defined as:
\[
d_t^{(a-b)} = L\left(e_t^{(a)}\right) - L\left(e_t^{(b)}\right),
\]
where \( e_t^{(a)} \) (\( e_t^{(b)} \)) represents the forecast error of model \( a \) (model \( b \)), respectively.\footnote{The null hypothesis of the DM test is \( \mathbb{E}\left(d_t^{(a-b)}\right) = 0 \) (i.e., no significant difference in forecast accuracy between the two models). Under the covariance stationary assumption, the test statistic converges to a standard normal distribution:
\[
\mathrm{DM}_{a,b} = \frac{\bar{d}^{(a-b)}}{\widehat{\sigma}^{(a-b)}} \to N(0,1),
\]
where \( \bar{d}^{(a-b)} = \frac{1}{T}\sum_{t=1}^T d_t^{(a-b)} \) is the sample mean of the loss difference \( d_t^{(a-b)} \), and \( \widehat{\sigma}^{(a-b)} \) is a consistent estimate of the standard deviation of \( \bar{d}^{(a-b)} \).}
We adopt the one-sided version of the DM test to evaluate whether a model generates significantly better (rather than merely significantly different) forecasts. Specifically, we test the null hypothesis that model \( a \)'s forecast loss is less than or equal to that of model \( b \); rejecting this null hypothesis implies that model \( b \) yields significantly superior forecasts.

\subsection{In-sample results}
\noindent The in‐sample regression analysis of the different HAR(X) specifications
and simple regressions covers the whole data set from January 2011 to December 2019.
Table \ref{tab:ins fit} presents the in-sample OLS estimation results for 1-day-ahead volatility forecasting models, including the benchmark HAR-RV and five extended specifications (HAR-RV-VIX, HAR-RV-Heston, HAR-RV-Bates, HAR-RV-SVCJ, HAR-RV-RHeston). Heteroskedasticity-robust standard errors \citep{newey1987hypothesis} are reported in parentheses below coefficient estimates, and adjusted \(R^2\) values quantify each model’s in-sample explanatory power.

For the benchmark HAR-RV model, all three horizon-specific volatility components (\(\beta_d\), \(\beta_w\), \(\beta_m\)) are positive and statistically significant, reflecting volatility’s heterogeneous persistence across daily, weekly, and monthly horizons \citep{corsi2009simple}. 
The model achieves an adjusted \(R^2\) of 0.605, establishing a baseline for evaluating extended specifications.
Incorporating the VIX into the framework (HAR-RV-VIX) improves explanatory power (adjusted \(R^2 = 0.627\)). The VIX coefficient is estimated at 0.6772 and highly significant (t-statistic = 10.166), confirming option-based volatility’s incremental gains for explaining realized volatility \citep{bekaert2014vix,todorov2022information,michael2025options,ruan2025merton}. Notably, the monthly component (\(\beta_m\)) becomes insignificant in this model, suggesting VIX absorbs information from longer-term historical realized volatility.

For models integrating stochastic volatility (SV) components, the HAR-RV-Heston specification (adjusted \(R^2 = 0.624\)) includes a statistically significant SV-Heston coefficient (0.5478, t-statistic = 9.469), with two of original HAR components (daily and weekly) remaining positive and significant. The HAR-RV-Bates (adjusted \(R^2 = 0.625\)) and HAR-RV-SVCJ (adjusted \(R^2 = 0.628\)) models follow a similar pattern: their respective Bates (0.5671) and SVCJ (0.5354) coefficients are significant, and the HAR components (daily and weekly) retain their statistical significance. Among all specifications, the HAR-RV-RHeston model outperforms its counterparts, achieving the highest adjusted \(R^2 = 0.644\). Its SV-RHeston coefficient (0.657, t-statistic = 13.767) is both economically meaningful and statistically robust, which highlights the rough volatility component’s capacity to capture refined volatility dynamics \citep{gatheral2018volatility,chong2025nonparametric}.

Across all extended models, the adjusted \(R^2\) values (0.624–0.644) exceed the benchmark HAR-RV’s 0.605, demonstrating that adding either VIX or SV factors enhances in-sample explanatory power. The hierarchical increase in adjusted \(R^2\) (peaking at HAR-RV-RHeston) further indicates that the rough Heston component provides the most incremental information for modeling one-day-ahead realized volatility.

\begin{table}[htbp]
\caption{\textbf{In‐sample 1-day forecast evaluation}}
\par
{\footnotesize
This table reports the 1-day OLS results. 
Values in parentheses indicate the heteroskedasticity‐robust standard errors based on the \cite{newey1987hypothesis} covariance correction. 
}
\noindent
\begin{center}
{\scriptsize   
\begin{tabular}{lcccccc}
\toprule
Model          & HAR-RV   & HAR-RV-VIX & HAR-RV-Heston & HAR-RV-Bates & HAR-RV-SVCJ & HAR-RV-RHeston \\
\midrule
intercept      & -0.2264  & -0.0785    & -0.1572       & -0.1500      & -0.2375     & -0.2576        \\
               & (-4.146) & (-1.427)    & (-2.926)      & (-2.789)     & (-4.482)    & (-4.967)       \\
$\beta_d$      & 0.4624   & 0.3792     & 0.3884        & 0.3946       & 0.4757      & 0.3105         \\
               & (16.651) & (13.45)    & (13.782)      & (13.62)      & (13.379)    & (10.871)       \\
$\beta_w$      & 0.2572   & 0.157      & 0.1642        & 0.168        & 0.1552      & 0.0396         \\
               & (0.6088) & (3.719)    & (3.878)       & (3.981)      & (3.681)     & (2.898)        \\
$\beta_m$      & 0.2038   & -0.0135    & 2.45          & 0.0121       & 0.0074      & 0.0328         \\
               & (5.471)  & (-0.321)    & (0.598)       & (0.292)      & (0.181)     & (0.0874)       \\
VIX            &          & 0.6772     &               &              &             &                \\
               &          & (10.166)   &               &              &             &                \\
SV-Heston      &          &            & 0.5478        &              &             &                \\
               &          &            & (9.469)       &              &             &                \\
SV-Bates       &          &            &               & 0.5671       &             &                \\
               &          &            &               & (9.619)      &             &                \\
SV-SVCJ        &          &            &               &              & 0.5354      &                \\
               &          &            &               &              & (10.333)    &                \\
SV-RHeston     &          &            &               &              &             & 0.657          \\
               &          &            &               &              &             & (13.767)       \\
Adj \(R^2\)    & 0.605    & 0.627      & 0.624         & 0.625        & 0.628       & 0.644          \\
\bottomrule
\end{tabular}        
}
\end{center}
\label{tab:ins fit}
\end{table}

\subsection{Out-of-sample}
\noindent Comparing the out‐of‐sample forecast performance of the different forecasting methods is even more important and
informative for analyzing the information content of the spot volatility estimators inferred from stochastic volatility models, especially rough Heston model. Our out-of-sample period is from Januray 2020 to June 2020 used the predictive regression model trained from the whole in-sample period.
Still, in addition to the HAR(X) models,
and simple regressions, we also use the non-parametric option-based volatility estimator VIX as benchmarks for the forecasts.
Table \ref{tab:oos fit} lists the one-day-ahead forecast accuracy for the candidate models across several evaluation metrics including MAE, RMSE, QLIKE, and MDA as described in Section \ref{sec:forecast_evaluation}.

First, all extended models outperform the benchmark HAR-RV across every metric. The HAR-RV-RHeston model achieves the lowest MAE (0.2137), RMSE (0.2762), and QLIKE (0.0403). It realizes a 9\% reduction in MAE and a 5.0\% reduction in RMSE relative to HAR-RV. This confirms that the rough Heston component captures incremental informational content of volatility dynamics, which in turn improves predictive accuracy. The HAR-RV-VIX model also outperforms HAR-RV, with MAE of 0.2205 and RMSE of 0.2798. This result is consistent with the in-sample finding that VIX contains valuable option-implied information for one-day-ahead volatility forecasting.
Second, the rough Heston model dominates in directional accuracy. The HAR-RV-RHeston model achieves the highest MDA at 73.60\%. This metric indicates the model correctly predicts the direction of volatility changes (rise or fall) in 73.60\% of trading days, which is 5.6 percentage points higher than the HAR-RV benchmark. This performance is particularly meaningful for practical applications such as trading strategy optimization, where directional signals are often as valuable as level forecasts.
Third, stochastic volatility (SV) models show heterogeneous performance. Models with standard SV components, including HAR-RV-Heston, HAR-RV-Bates, and HAR-RV-SVCJ, deliver similar accuracy. Their MAE values range from 0.2190 to 0.2221, and RMSE values range from 0.2776 to 0.2815. None of these models, however, match the performance of the rough Heston model. This outcome highlights that incorporating rough volatility dynamics provides unique predictive value for short-term volatility, distinct from the value offered by traditional SV structures.

\begin{table}[htbp]
\caption{\textbf{Ouf-of‐sample 1-day forecast evaluation}}
\par
{\footnotesize
This table reports the 1-day OLS results. 
Values in parentheses indicate the heteroskedasticity‐robust standard errors based on the \cite{newey1987hypothesis} covariance correction. 
}
\noindent
\begin{center}
{\scriptsize   
\begin{tabular}{rrrrrrr}
\toprule
Model  & HAR-RV   & HAR-RV-VIX & HAR-RV-Heston & HAR-RV-Bates & HAR-RV-SVCJ & HAR-RV-RHeston \\
\midrule
MAE    & 0.2351   & 0.2205     & 0.2221        & 0.2199       & 0.2190      & 0.2137         \\
RMSE   & 0.2904   & 0.2798     & 0.2815        & 0.2788       & 0.2776      & 0.2762         \\
QLIKE  & 0.0428   & 0.0409     & 0.0414        & 0.0406       & 0.0404      & 0.0403         \\
MDA    & 68.00\%  & 71.20\%    & 70.40\%       & 68.80\%      & 70.40\%     & 73.60\%        \\
\bottomrule
\end{tabular}         
}
\end{center}
\label{tab:oos fit}
\end{table}

Apart from average forecast performance metrics, we also examine the distribution of forecast errors (via MAE and QLIKE) to assess forecast accuracy across the test period, as shown in Figure \ref{fig:boxplot}.
This figure presents boxplots of error metrics, where each box depicts the median (solid black line), 25th percentile (Q1), and 75th percentile (Q3) of errors for each model. 
For MAE errors (left panel): The HAR-RV-RHeston model exhibits the smallest median error, although the tightness of the interquartile range (IQR, i.e., the height of the box) across all candidate models are quite similar, indicating similar error stability. 
In contrast, the benchmark HAR-RV has the largest median error.
For QLIKE errors (right panel): The pattern mirrors MAE: HAR-RV-RHeston has the smallest median error.

\begin{figure}[htbp]
\begin{center}
\includegraphics[width=0.80\textwidth]{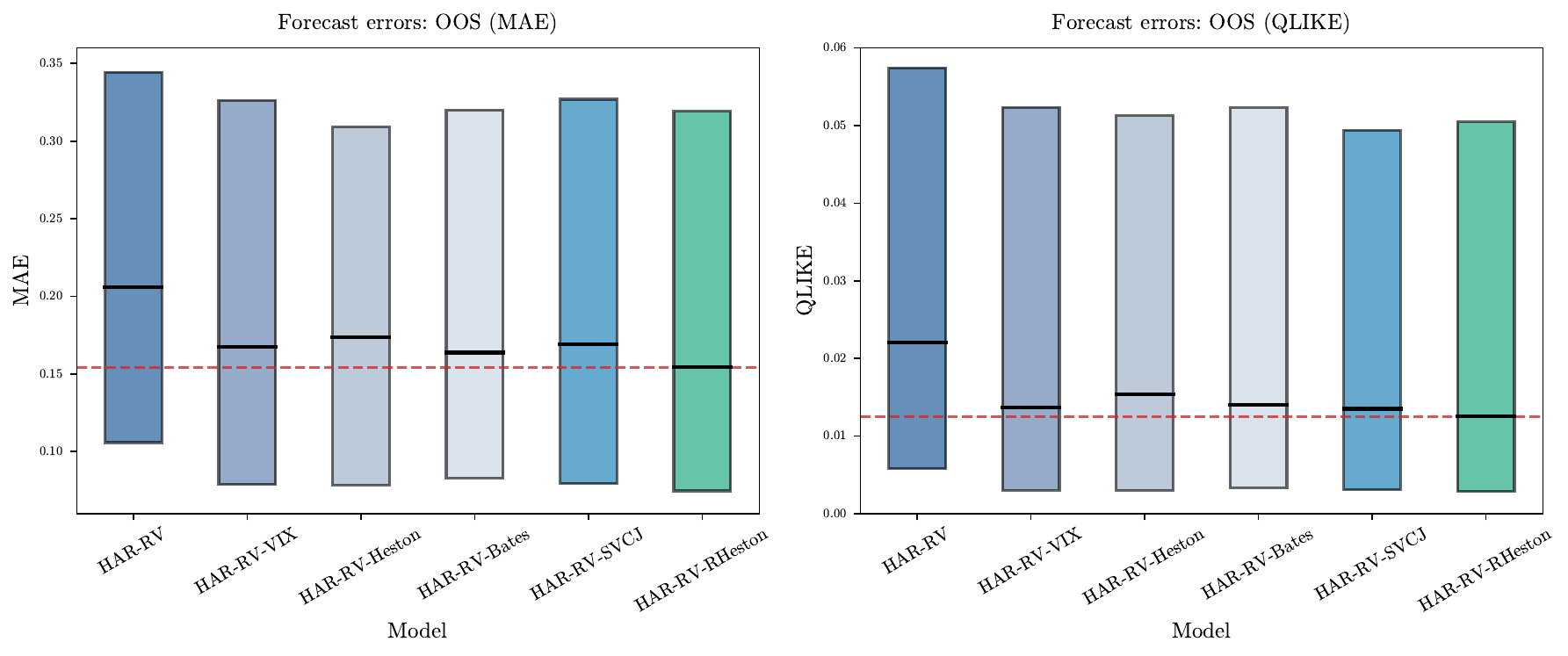}
\end{center}
\caption{\textbf{Boxplots of forecast errors for models}}
{
\footnotesize
\noindent
\begin{spacing}{1.4}
This figure presents boxplots illustrating three summary statistics: the median, and the Q1 and Q3 quantiles. 
\end{spacing}
}
\label{fig:boxplot}
\end{figure}

Table \ref{tab:DM} presents the results of pairwise Diebold-Mariano (DM) tests \citep{diebold2002comparing} for 1-day-ahead volatility forecasts. The DM test is a standard tool for statistically comparing forecast accuracy, where a negative test statistic in cell \((i,j)\) indicates that model \(i\) outperforms model \(j\) at different significance levels. 
As shown in the first column of Table \ref{tab:DM}, all extended models generate negative DM statistics relative to the benchmark HAR-RV: for instance, HAR-RV-RHeston yields a statistic of -3.1179, HAR-RV-Bates of -3.2297, and HAR-RV-SVCJ of -3.259, confirming that incorporating either VIX or stochastic volatility (SV) components into the HAR framework leads to statistically significant improvements in forecast performance. 
Notably, the HAR-RV-RHeston model exhibits significantly negative statistics against all other extended specifications: it outperforms HAR-RV-VIX (-2.5446), HAR-RV-Heston (-2.5907), HAR-RV-Bates (-1.7681), and HAR-RV-SVCJ (-1.8359), demonstrating that the rough Heston component provides unique, statistically meaningful incremental value beyond both non-parametric option-implied information (VIX) and traditional SV structures. Among other extended models, HAR-RV-SVCJ outperforms HAR-RV-Heston (-1.6248) and HAR-RV-Bates (-0.4571), while HAR-RV-Bates outperforms HAR-RV-Heston (-1.1208); notably, HAR-RV-Heston produces a positive statistic (1.193) when compared to HAR-RV-VIX, suggesting that VIX’s option-implied information may be more valuable for short-term forecasting than the standard Heston component.
Collectively, these DM test results validate the robustness of the HAR-RV-RHeston model’s superior performance—its predictive advantage over both the benchmark and other extended models is statistically significant, reinforcing the value of integrating rough volatility dynamics into volatility forecasting frameworks. In summary, the out-of-sample results reinforce the in-sample findings. Extending the HAR framework with either VIX or SV components improves forecast accuracy. Among all specifications, the rough Heston model is the most effective.

\begin{table}[htbp]
\caption{\textbf{Pairwise comparisons: forecast accuracy}}
\par
{\footnotesize
This table shows the \cite{diebold2002comparing} test statistics of the pricing errors ($\text{MAE}$) differences. A negative (positive) statistic in a cell $(i, j)$ indicates that the model $i$ outperforms (underperforms) the model $j$. 
}
\noindent
\begin{center}
{\scriptsize   
      \begin{tabular}{rrrrrr}  
\toprule
Model  
& HAR-RV   & HAR-RV-VIX & HAR-RV-Heston 
& HAR-RV-Bates & HAR-RV-SVCJ \\
\midrule
HAR-RV 
& -        &            &               &              &             \\
HAR-RV-VIX     & -2.9456***& -          &               &              &             \\
HAR-RV-Heston  & -2.7095***& 1.193      & -             &              &             \\
HAR-RV-Bates   & -3.257*** & -0.3013    & -1.1208       & -            &             \\
HAR-RV-SVCJ    & -3.2299***& -1.0397    & -1.6248*      & -0.4571      & -           \\
HAR-RV-RHeston & -3.1179***& -2.5446*** & -2.5907***    & -1.7681**    & -1.8359**   \\
\bottomrule
\end{tabular}
}
\end{center}
\label{tab:DM}
\end{table}

\subsection{Multiple horizons of forecasting power}

\noindent One-day-ahead volatility forecasting is rarely the sole focus for market practitioners. Multi-horizon predictions (e.g., weekly or monthly) sometimes can be equally critical for algorithmic trading strategy optimization and risk management. Following the convention in prior literature \citep{zhang2025forecasting}, we extend our analysis to {$h$-day-ahead forecasts} (where $h = 1, 2, \dots, 21$), with the target volatility defined as the cumulative realized volatility over the horizon: 
$
RV_{t: t+h} = \sum_{k=0}^h RV^d_{t+k}
$

We evaluate model performance across 22 horizons using four metrics (MAE, QLIKE, RMSE, and MDA; see Section \ref{sec:forecast_evaluation}), benchmarking the proposed HAR-RV-RHeston against two baselines: the original HAR-RV, and HAR-RV-VIX (which incorporates a non-parametric option-implied volatility estimator). The results are visualized in Figure \ref{fig:horizons}.

As evident from the plots, for {error metrics (MAE, QLIKE, RMSE)}, HAR-RV-RHeston consistently yields \textit{lower values} across all horizons (1--22 days) compared to both benchmarks; even in short horizons (e.g., $h \leq 5$), where all models show a sharp initial drop in error, HAR-RV-RHeston maintains a clear advantage, and this outperformance persists (and in some cases widens) as the forecasting horizon extends. For {directional accuracy (MDA)}, HAR-RV-RHeston delivers {higher MDA values} for most horizons, reflecting more reliable predictions of volatility's upward and downward trends, though it exhibits minor dips in MDA for a few mid-horizon cases (e.g., $10 \leq h \leq 15$), with its overall directional performance still outpacing the two baseline models.
In sum, the spot volatility estimated from the rough Heston model provides stronger explanatory and predictive power than the VIX-based counterpart \textit{across nearly all forecasting horizons}---as corroborated by the lower error and (mostly) higher directional accuracy of HAR-RV-RHeston.

\begin{figure}[htbp]
\begin{center}
\includegraphics[width=1.00\textwidth]{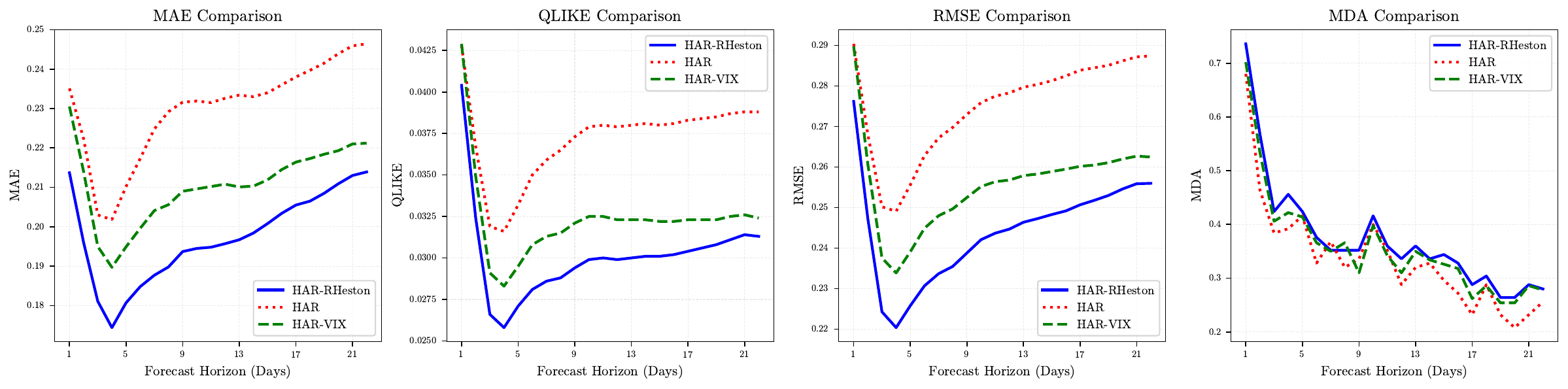}
\end{center}
\caption{\textbf{Multi-horizon forecasting performance}}
{
\footnotesize
\noindent
\begin{spacing}{1.4}
This figure compares the performance of HAR-RV-RHeston (blue), HAR-RV (red dashed), and HAR-RV-VIX (green dashed) across 1--22 day-ahead forecasting horizons, using MAE, QLIKE, RMSE, and MDA. The target volatility for horizon $h$ is defined as the cumulative realized volatility $RV_{t:t+h} = \sum_{k=0}^h RV^d_{t+k}$. 
\end{spacing}
}
\label{fig:horizons}
\end{figure}

%-----Conclusions-------------------------
%-----------------------------------------
\subsection{Future research}
\noindent Rough volatility models hold profound implications for asset pricing \citep{bouchaud2020perfect}. A key recent advancement is the quadratic rough Heston model \citep{gatheral2020quadratic}, which resolves a long-standing challenge in quantitative finance: it enables the joint calibration of S\&P 500 and VIX option smiles, an achievement that eluded practitioners previously due to conflicting market dynamics.
Given this research letter’s focus on documenting rough volatility’s informational gains, we confined our analysis to the standard rough Heston model, which offers moderate analytical tractability. This choice stems from computational constraints: the quadratic rough Heston model’s greater complexity, combined with our long-span option panel, incurs substantial computational costs even with advanced deep learning techniques \citep{rosenbaum2021deep}. The standard specification aligns with our core goal of isolating rough volatility’s incremental informational value.
We anticipate extending this framework to the quadratic rough Heston model will further enhance realized volatility forecasting. Future research could prioritize optimizing computational efficiency for long-span datasets (e.g., via rough volatility-tailored accelerated numerical schemes) to fully exploit the quadratic specification’s superior calibration capabilities. Such an extension may deepen insights into rough volatility dynamics and their practical relevance for asset pricing and risk management.

\section{Concluding remarks}
\label{section-conclusion}
\noindent 
This paper explores the informational gains of rough volatility dynamics for realized volatility forecasting by augmenting the Heterogeneous Autoregressive (HAR) model with model-based spot volatility estimators extracted from traded options data. We address the computational challenges of rough volatility models via a deep learning surrogate, and systematically compare the performance of the rough Heston model against traditional stochastic volatility models (Heston, Bates, SVCJ) and the non-parametric VIX index.

Our key findings are threefold. First, the augmented HAR-RV-RHeston model outperforms all benchmarks in-sample and out-of-sample: it achieves the highest adjusted \(R^2\) (0.644) in-sample and the lowest out-of-sample forecast errors (MAE=0.2137, RMSE=0.2762, QLIKE=0.0403), with significant error reductions relative to the baseline HAR-RV. Second, it delivers superior directional accuracy (MDA=73.60\%), 5.6 percentage points above HAR-RV, underscoring its practical value for trading strategies and risk management. Third, its outperformance persists across 1–22 day multi-horizon forecasts, confirming the robustness of rough volatility dynamics.
Methodologically, we contribute to volatility forecasting literature by integrating rough volatility into the HAR framework, pairing \cite{andersen2015parametric}’s parametric inference with deep learning to resolve computational constraints of rough volatility models. This approach enriches the forecasting information set and provides a feasible solution for large-scale option panels. Empirically, we validate that option-implied spot volatility from rough Heston models contains unique informational value beyond traditional stochastic volatility components and the non-parametric VIX, highlighting the critical role of volatility roughness.

%----------------------------------------------------------------------------------
% \clearpage
% \nocite{}
% \small
% \bibliographystyle{jfe}
% \begingroup
% \setstretch{1.2}
% %\setlength\bibitemsep{0pt}
% \bibliography{ref}
% \endgroup
% \newpage
\nocite{}
\small
\bibliography{ref}
%\clearpage

% Figure
\newpage
~\\

\noindent \textbf{Acknowledgement}

\noindent 
Meng Melody Wang would like to extend her deepest gratitude to Dr. Martin Forde for his attentive supervision during her postgraduate studies at King's College London, which helped her gain access to the research domain of deep calibration and rough volatility. 
% Some of the ideas presented in this paper were developed during that period.
\\

\noindent
\textbf{Disclaimer}

\noindent
The views expressed here are those of the authors and not necessarily those of Green Pine Capital and Green Pine Innorise Fund.\\

\noindent \textbf{Data availability}\\
\noindent 
Data will be made available on request.\\

\noindent \textbf{Conflict of Interest Statement}\\
\noindent 
There is no conflict of interest to declare.\\

\noindent \textbf{Funding Information}\\
\noindent 
This work was supported by the Beijing Normal-Hong Kong Baptist University start-up research fund under Grant UICR0700136-26.
%--------------------------------------
\newpage
%--------------------------------------
\begin{appendices}
\section{Stochastic volatility model settings}
\label{sec:appendix SV}
\subsection{Formulation of affine stochastic volatility models with jumps}
\label{appendix:model details}
%------------------------------
\noindent Under the class of affine jump-diffusion models, the joint process of $S_t$ and $V_t$ under a risk neutral measure $\mathbb{Q}$ is specified by the dynamic equations:

\begin{equation}
\begin{aligned}
\frac{\dd S_t}{S_t} & =(r-\eta \mu_S) \dd t+\sqrt{V_t} \dd W_t^1+\left(e^{J^S}-1\right) \dd N_t \\
\dd V_t & =\lambda\left(\theta-V_t\right) \dd t+\nu \sqrt{V_t} \dd W_t^2+J^V \dd N_t,
\end{aligned}
\end{equation}
where the Brownian motions $W_t^1$ and $W_t^2$ observe $\dd W_t^1 \dd W_t^2=\rho \dd t$. Here, $\rho$ is the constant correlation coefficient, $\theta$ is the constant mean reversion level of $V_t, \nu$ is the constant volatility of $V_t$ and $\lambda$ is the constant mean reversion speed. We assume simultaneous jumps on $S_t$ and $V_t$ and they are modeled by the common Poisson process $N_t$ with constant intensity $\eta$. 
Let $J^S$ and $J^V$ denote the respective random jump component on $S_t$ and $V_t$, where $J^S$ and $J^V$ are independent of $N_t$ and both random jump components are contemporaneous. Furthermore, we assume $J^V$ to be exponentially distributed with mean $\mu_V$, 
$
J^V \sim \exp \left(\mu_V\right),
$
and $J^S$ is normally distributed 
$J^S  \sim N\left(\mu_S, \sigma_S^2\right).$
The full specification is called the SVCJ model. By removing the jump component in the variance process, we obtain the Bates model. The details of the model features are summarized in Table \ref{tab:models}. The top panel shows which model features are included in the nested models. The checkmark indicates that the model includes the corresponding features. The bottom panel shows which model parameters are excluded.  

\begin{table}[htbp]
\caption{\textbf{Comparison of model features of RHeston model and other three Heston-type models}}
\par
{\footnotesize

}
\noindent
\begin{center}
{\scriptsize   
\begin{tabular}{ccccc}
\toprule
           &     Heston &      Bates &   SVCJ &    RHeston \\
\midrule
                           \multicolumn{ 5}{c}{Model features} \\
\midrule
rough variance &            &            &            &          $\checkmark$ \\

return jump &            &          $\checkmark$ &          $\checkmark$ &            \\

variance jump &            &            &          $\checkmark$ &            \\
\midrule
                \multicolumn{ 5}{c}{Excluded model parameters} \\
\midrule
         $H$ &          \ding{55} &          \ding{55} &          \ding{55} &            \\

      $\eta$ &          \ding{55} &            &            &          \ding{55} \\

     $\mu_S$ &          \ding{55} &            &            &          \ding{55} \\

     $\mu_V$ &          \ding{55} &          \ding{55} &            &          \ding{55} \\
\bottomrule
\end{tabular}          
}
\end{center}
\label{tab:models}
\end{table}
% \subsubsection{Lifted Heston: An efficient solution for rough Heston}

\section{Options data}
\label{sec:options data}
\noindent 
Our SPX options datasets comprise the put and call options prices and T-bill rates sourced from OptionMetrics. 
The sample periods are the same as the return dataset used for computing realized volatilities. We implement several filters before empirical analysis mainly for eliminating inaccurate and illiquid options following \citep{Bakshi1997,christoffersen2009shape} among others.  
We restrict the analysis to time to maturities between one week and one year.
We delete options that report null or zero open interests, trading volumes, implied volatility, and vega on a given date. We delete options for which the mid-price is too small or the relative bid-ask spread is larger than 0.3.
Option prices are taken from the bid-ask midpoint at each day’s close of the options market. 
We proxy for the riskless interest rate by using the daily available T-bill rates, which we interpolate to match the maturities of the options contracts that we observe.
We also discard the options that violate the model-free lower bound constraints:
\begin{equation}
\begin{aligned}
& \widehat{C}(t, T, K) \geq \max \left(0, S-K\right) \\
& \widehat{P}(t, T, K) \geq \max \left(0, K-S\right)
\end{aligned}
\end{equation}
where $\widehat{C}(t, T, K),\widehat{P}(t, T, K)$ denote the options at date $t$ with maturity $T$.
Moreover, OTM options tend to be more liquid than in-the-money ones. For this reason, we use only OTM call and put options. 
Table \ref{tab:options data} shows summary statistics such as the number of observations and the average implied volatilities. In this table, the data are firstly divided into several categories depending on time to maturity $\tau$, we further break down the data into several categories according to their moneyness $K / S$.

\begin{table}[htbp]
\caption{\textbf{Summary statistics for SPX options data}}
\par
{\footnotesize
This table shows the descriptive statistics of SPX options prices obtained from OptionMetrics. 
}
\noindent
\begin{center}
{\scriptsize   
\begin{tabular}{rrrrrrrrrrr}
\toprule
      $\tau$ (days) &            &  $7<\tau<45$ &            &  $45<\tau<90$ &            & $90<\tau<180$ &            & $180<\tau<360$ &            &        All \\
      \midrule
 \multicolumn{ 11}{c}{Panel A: Number of contracts} \\
\midrule
$K/S < 0.975$ &            &     137141 &            &      76515 &            &      48478 &            &      35562 &            &     357545 \\

$0.975 <= K/S < 1$ &            &     296478 &            &     156179 &            &      74783 &            &      32226 &            &     681286 \\

$1 <= K/S < 1.025$ &            &     398455 &            &     175466 &            &      72064 &            &      29874 &            &     852323 \\

$K/S >= 1.025$ &            &     446602 &            &     319037 &            &     171075 &            &      95675 &            &    1164108 \\

       All &            &    1278676 &            &     727197 &            &     366400 &            &     193337 &            &    3055262 \\

       \midrule
  \multicolumn{ 11}{c}{Panel B: Average IV} \\
\midrule
X/F $<$ 0.975 &            &    0.2249  &            &    0.2202  &            &    0.2326  &            &    0.2206  &            &    0.2296  \\

0.975 $<=$ X/F $<$ 1 &            &    0.1661  &            &    0.1702  &            &    0.1834  &            &    0.1865  &            &    0.1710  \\

1 $<=$ X/F $<$ 1.025 &            &    0.1389  &            &    0.1478  &            &    0.1706  &            &    0.1795  &            &    0.1461  \\

X/F $>=$ 1.025 &            &    0.1625  &            &    0.1460  &            &    0.1563  &            &    0.1585  &            &    0.1627  \\

       All &            &    0.1627  &            &    0.1595  &            &    0.1748  &            &    0.1779  &            &    0.1678  \\
\bottomrule
\end{tabular}      
}
\end{center}
\label{tab:options data}
\end{table}

\section{Analytical option pricing and fast Fourier transform}
\label{sec:ChF}
\begin{lemma}[Characteristic function under the rough Heston model]

\begin{equation}
\mathbb{E}\left[e^{\i a X_t}\right]=\exp \left(g_1(a, t)+V_0 g_2(a, t)\right),
\end{equation}
where
$$
g_1(a, t)=\theta \lambda \int_0^t h(a, s) \dd s, \quad g_2(a, t)=I^{1-\alpha} h(a, t),
$$
and $h$ is a solution of the following fractional Riccati equation:
\begin{equation}
D^\alpha h=\frac{1}{2}\left(-a^2-\i  a\right)+\lambda(\i a \rho \nu-1) h(a, s)+\frac{(\lambda \nu)^2}{2} h^2(a, s), \quad I^{1-\alpha} h(a, 0)=0,
\end{equation}
with $D^\alpha$ and $I^{1-\alpha}$, for $\alpha \in(0,1]$, the fractional derivative and integral operators defined as
$$
I^\alpha f(t)=\frac{1}{\Gamma(\alpha)} \int_0^t(t-s)^{\alpha-1} f(s) \dd s, \quad
D^\alpha f(t)=\frac{1}{\Gamma(1-\alpha)} \frac{\dd}{\dd t} \int_0^t(t-s)^{-\alpha} f(s) \dd s.
$$
% and
% $$
% D^r f(t)=\frac{1}{\Gamma(1-r)} \frac{\dd}{\dd t} \int_0^t(t-s)^{-r} f(s) \dd s.
% $$
Remark that when $\alpha=1$, this result indeed coincides with the classical Heston's result. However, note that for $\alpha<1$, the solutions of such Riccati equations are no longer explicit. Nevertheless, they are easily solved numerically via special computational methods.
\end{lemma}
\begin{proof}
See \cite{el2019characteristic}.
\end{proof}

\begin{lemma}[Characteristic function under the lifted Heston model]
\begin{equation}
\mathbb{E}\left[\exp \left(u \log S_t^n\right) \mid \mathcal{F}_t\right]=\exp \left(\phi^n(t, T)+u \log S_t^n+\sum_{i=1}^n c_i^n \psi^{n, i}(T-t) U_t^{n, i}\right)
\end{equation}
\end{lemma}
\begin{proof}
See \cite{AbiJaber2019}.
\end{proof}

~

\noindent A European option with maturity $T$ and strike $K$ is priced as the risk-free discounted expected payoff under the risk-neutral measure $\mathbb{Q}$. Specifically, the time-$t$ price of a European call option is given by
\begin{equation}
C\left(S_t, t, v_t\right) \equiv \mathbb{E}^{\mathbb{Q}}\left[e^{-r(T-t)}\left(S_T-K\right)^{+}\right],
\end{equation}
where $(x)^+=\max(x,0)$ denotes the positive part operator.
 This price can be efficiently derived via Fourier inversion, following the closed-form representation:
\begin{equation}
\label{eq:fourier}
C_t=e^{-q (T-t)} S_t \Pi_1-e^{-r (T-t)} K \Pi_2,
\end{equation}
with $S_t$ the underlying spot price at time $t$, $r$ the constant risk-free rate, $q$ the continuous dividend yield, and $T-t$ the time to expiration. The terms $\Pi_1$ and $\Pi_2$ are probability-related integrals defined as
\begin{equation}
\begin{aligned}
\Pi_1&=\frac{1}{2}+\frac{1}{\pi} \int_0^{\infty} \operatorname{Re}\left(\frac{e^{-\mathrm{i} \omega \log(X)} \phi(\omega-\mathrm{i})}{\mathrm{i} \omega \phi(-\mathrm{i})}\right) \mathrm{d} \omega, \\
\Pi_2&=\frac{1}{2}+\frac{1}{\pi} \int_0^{\infty} \operatorname{Re}\left(\frac{e^{-\mathrm{i} \omega \log(X)} \phi(\omega)}{\mathrm{i} \omega}\right) \mathrm{d} \omega,
\end{aligned}
\end{equation}
where $\phi$ stands for the characteristic function of the log stock price under a stochastic volatility model, such as rough Heston \eqref{eq-rheston}; the function $\operatorname{Re}(\cdot)$ returns the real part of a complex number. 
Given an explicit form of $\phi(\cdot)$, $\Pi_1$ and $\Pi_2$ are computable via numerical integration, which in turn yields the call price through \eqref{eq:fourier}.\footnote{We focus on call option pricing hereafter, as put prices can be readily recovered via put-call parity.}

% ~

\end{appendices}
\end{document}